\documentclass{llncs}
\usepackage{epsfig}
\usepackage{gensymb}
\usepackage{subfigure}
\usepackage[english]{babel}
\usepackage{graphicx}
\usepackage{amssymb}
\usepackage{multirow}
\usepackage{enumerate}

\usepackage{amsmath,amscd,amssymb}
\usepackage{mathrsfs}
\usepackage{graphicx}

\usepackage[T1]{fontenc} 
\usepackage[ansinew]{inputenc} 

	\newcommand{\RR}{{\mathbb R}}

\def\cn{\mbox{\rm cr}}

\def\lcr{\mbox{\rm $\overline{\textnormal{\rm cr}}$}}

\def\cn{\mbox{\rm cr}}

\def\rk{\mbox{\rm Rank}}
\def\sgn{\mbox{\rm sgn}}

\renewenvironment{proof}
{{\bf Proof:}}{\hspace*{\fill}$\Box$\par\vspace{1mm}}

\def\cn{\mbox{\rm cr}}

\title{Approximating the rectilinear crossing number}

\author{Jacob Fox\inst{1}\thanks{Supported by a Packard Fellowship, by NSF CAREER award DMS 1352121, and by an Alfred P. Sloan Fellowship.} \and J\'anos Pach\inst{2}\thanks{Supported
by a Hungarian Science Foundation NKFI grant, by Swiss National Science Foundation Grants 200021-165977 and 200020-162884.}\and Andrew Suk\inst{3}\thanks{Supported by NSF grant DMS-1500153.}}

\institute{Stanford University, Stanford, CA, USA\\{\tt jacobfox@stanford.edu}\and EPFL, Lausanne, Switzerland and Courant Institute, New York, NY, USA\\ {\tt pach@cims.nyu.edu} \and University of Illinois at Chicago, Chicago, IL, USA\\ \texttt{suk@uic.edu}}

\begin{document}

\maketitle

\begin{abstract}
A \emph{straight-line} drawing of a graph $G$ is a mapping which assigns to each vertex a point in the plane and to each edge a straight-line segment connecting the corresponding two points.  The \emph{rectilinear crossing number} of a graph $G$, $\lcr(G)$, is the minimum number of pairs of crossing edges in any straight-line drawing of $G$.  Determining or estimating $\lcr(G)$ appears to be a difficult problem, and deciding if $\lcr(G)\leq k$ is known to be NP-hard.  In fact, the asymptotic behavior of $\lcr(K_n)$ is still unknown.

In this paper, we present a deterministic $n^{2+o(1)}$-time algorithm that finds a straight-line drawing of any $n$-vertex graph $G$ with $\lcr(G) + o(n^4)$ pairs of crossing edges.  Together with the well-known Crossing Lemma due to Ajtai et al.~and Leighton, this result implies that for any dense $n$-vertex graph $G$, one can efficiently find a straight-line drawing of $G$ with $(1 + o(1))\lcr(G)$ pairs of  crossing edges.

\end{abstract}


\section{Introduction}

A \emph{drawing} of a graph $G$ is a mapping $f$ that assigns to each vertex a distinct point in the plane and to each edge $uv$ a continuous arc connecting $f(u)$ and $f(v)$, not passing through the image of any other vertex.  Two edges in a drawing \emph{cross} if their interiors have a point in common.  The \emph{crossing number} of $G$, denoted by $\cn(G)$, is the minimum number of pairs of crossing edges in any drawing of $G$.  Hence, $\cn(G) = 0$ if and only if $G$ is planar.  Determining or estimating the crossing number of a graph is one of the oldest problems in graph theory, with over 700 papers written on the subject.  We refrain here from attempting to give an overview of the long history of crossing numbers and their applications in discrete and computational geometry, and refer the reader to the survey articles by Pach and T\'oth \cite{PT}, Schaefer \cite{S2}, and the extensive bibliography maintained by Vrt'o \cite{V}.

In the present paper, we focus on \emph{straight-line drawings} of a graph $G$, that is, drawings of $G$ where the edges are represented by straight-line segments.  We will assume that in all such drawings, no three vertices are collinear, and no point lies in the interior of three distinct edges.  The  \emph{rectilinear crossing number} of $G$, denoted by $\lcr(G)$, is the minimum number of pairs of crossing edges in any straight-line drawing of $G$.  Clearly $\cn(G) \leq \lcr(G)$, and a theorem of F\'ary \cite{F} states that $\lcr(G) = 0$ when $G$ is planar.  On the other hand, it was shown by Bienstock and Dean \cite{BD93} that there are graphs with crossing number four, whose rectilinear crossing numbers are arbitrarily large.

Determining the rectilinear crossing number of a graph appears to be a difficult problem.  In fact, the asymptotic value of $\lcr(K_n)$ is still unknown.  The exact values for $\lcr(K_n)$ are known for $n \leq 27$ and $n = 30$, and for large $n$, the current best known bounds are

$$0.379972{n\choose 4} < \lcr(K_n) < 0.380473{n\choose 4},$$

\noindent due to \'Abrego et al. \cite{Ab11} and Fabila-Monroy and L\'opez \cite{FML} respectively.  For more details on $\lcr(K_n)$, including its striking connection to Sylvester's four-point problem \cite{Syl64,Syl65}, see \cite{Ab,SW}.

From an algorithmic point of view, computing $\lcr(G)$ is known to be NP-hard \cite{B}.  More precisely, it is known to be $\exists \mathbb{R}$-complete, that is, complete for the existential theory of the reals (see \cite{S10,S2}).  On the other hand, many researchers have designed polynomial time algorithms for approximating crossing numbers of \emph{sparse} graphs. In particular, a seminal result of Hopcroft and Tarjan \cite{HT} is that there is a linear time algorithm for testing planarity of a graph.  Kawarabayashi and Reed \cite{KR} generalized their result and established a linear time algorithm that decides whether $\cn(G) \leq k$ when $k$ is fixed.  Leighton and Rao \cite{LR} obtained an efficient algorithm that finds a drawing of any bounded-degree $n$ vertex graph $G$ with at most $O(\log^4n)(n + \cn(G))$ pairs of crossing edges.  This was later improved by Even, Guha, and Schieber \cite{EGS} to $O(\log^3n)(n + \cn(G))$, and further improved by Arora, Rao, and Vazirani \cite{ARV} to $O(\log^2n)(n + \cn(G))$.  For more results on computing $\cn(G)$ for bounded degree graphs, see \cite{Ch}.

For \emph{dense} graphs $G$, very little is known about $\lcr(G)$, and as mentioned above, not even the asymptotic value of $\lcr(K_n)$.  Our main result is the following.

\begin{theorem}\label{main}

There is a deterministic $n^{2 + o(1)}$-time algorithm for constructing a straight-line drawing of any $n$-vertex graph $G$ in the plane with $$\lcr(G) + O(n^4/(\log\log n)^{\delta})$$ crossing pairs of edges, where $\delta>0$ is an absolute constant.

\end{theorem}

A classic result of Ajtai et al.~\cite{Aj} and Leighton \cite{L}, known as the Crossing Lemma, implies that the rectilinear crossing number of any $n$-vertex graph with $e$ edges is at least $\frac{e^3}{64n^2} - 4n$.  Hence all $n$-vertex graphs $G$ with $\Omega(n^2)$ edges satisfy $\lcr(G) \geq \Omega(n^4)$.  This implies the following.

\begin{corollary}

There is a deterministic $n^{2 + o(1)}$-time algorithm for constructing a straight-line drawing of any $n$-vertex graph $G$ with $|E(G)| \geq \varepsilon n^2$, where $\varepsilon > 0$ is fixed, such that the drawing has at most $(1 + o(1))\lcr(G)$ crossing pairs of edges.

\end{corollary}

A sequence $(G_n: n = 1, 2, \ldots)$ of graphs with $|V(G_n)| = n$ is called \emph{quasi-random with density} $p$ (where $0 < p < 1$) if, for all subsets $X,Y \subset V(G_n)$, $e_{G_n}(X,Y) = p|X||Y| + o(n^2)$. An important result of Chung, Graham, and Wilson \cite{CGW} shows that being quasi-random with density $p$ is equivalent to many other properties almost surely satisfied by the random graph $G(n,p)$. Studying properties of quasi-random graphs has been an important research direction with numerous applications. In Section~\ref{quasisect}, we prove the following result.

\begin{theorem}\label{quasithm}

Fix $0 < p < 1$ and let $(G_n : n  = 1, 2, \ldots)$ be a sequence of graphs that is quasi-random with density $p$.  Then

$$\lcr(G_n) = (1 + o(1))p^2\cdot\lcr(K_n).$$

\end{theorem}

\noindent More generally, we show any two edge-weighted graphs which are close in cut-distance have rectilinear crossing numbers which are close (see Lemma \ref{cross}). For results on crossing numbers of random graphs, consult \cite{ST}.

\medskip

\noindent \textbf{Organization.}  In the next section, we collect several geometric results on planar point sets and give an exponential time algorithm for computing the rectilinear crossing number of a (small) graph.  In Section \ref{sect2}, we show that if two graphs are close in cut-distance, then their rectilinear crossing numbers are approximately the same.  In Section \ref{algorithm}, we prove Theorem \ref{main}.   Finally in Section \ref{quasisect}, we prove Theorem \ref{quasithm}.

We omit floor and ceiling signs whenever they are not crucial. All logarithms are base 2.

\section{Order types and same-type transversals}

Let $V = (v_1,\ldots,v_n)$ be an $n$-element point sequence in $\mathbb{R}^2$ in general position, that is, no three members of $V$ are collinear.  The \emph{order type} of $V$ is the mapping $\chi:{V \choose  3}\rightarrow \{+1,-1\}$ (positive orientation, negative orientation), assigning each triple of $V$ its orientation.  By setting $v_i = (x_i,y_i) \in \mathbb{R}^2$, for $i_1  < i_2 < i_3$,

$$\chi(\{v_{i_1},v_{i_2},v_{i_3}\})  =  \sgn \hspace{.05cm} \det\left(\begin{array}{ccc}
                                                           1 & 1 &  1 \\
                                                           x_{i_1} & x_{i_2} & x_{i_3} \\
                                                           y_{i_1} & y_{i_2} & y_{i_3}
                                                         \end{array}\right).$$

\noindent Therefore, two $n$-element point sequences $V=(v_1,\ldots, v_n)$ and $U = (u_1,\ldots, u_n)$ have the same order type if they are ``combinatorially equivalent."   By lexicographically ranking each triple $(i_1,i_2,i_3)$, where $1 \leq i_1 < i_2 < i_3\leq n$, we can describe each order type $\chi$ with the vector $(\chi_1,\chi_2,\ldots) \in \{-1,+1\}^{n\choose 3}$, such that $\chi_j = +1$ if and only if $\chi(\{v_{i_1},v_{i_2},v_{i_3}\})  > 0$ and $\rk(i_1,i_2,i_3) = j$.  We will call vectors $\chi^{\ast} \in  \{-1,+1\}^{n\choose 3}$ \emph{abstract order types}, and we say that an abstract order type $\chi^{\ast}$ is \emph{realizable} if there is a point set $V$ in the plane whose order type realizes $\chi^{\ast}$.  The concept of order types was introduced by Goodman and Pollack \cite{GP83} and has played a crucial role in gathering knowledge about crossing numbers.  See \cite{GP83,GP93} for more background on order types.

Given $k$ disjoint subsets $V_1,\ldots, V_k \subset V$, a \emph{transversal} of $(V_1,\ldots,V_k)$ is any $k$-element sequence $(v_1,\ldots, v_k)$ such that $v_i \in V_i$ for all $i$.   We say that the $k$-tuple of parts $(V_1,\ldots, V_k)$  has \emph{same-type} transversals if all of its transversals have the same order type.  One of the key ingredients in the proof of Theorem \ref{main} is the following regularity lemma for same-type transversals established by the authors in \cite{FPS}.  A partition on a finite set $V$ is called {\em equitable} if any two parts differ in size by at most one.

\begin{theorem}\label{partition}

There is an absolute constant $C$ such that the following holds.  For each $0 < \varepsilon < 1$ and for any finite point set $V$ in $\RR^2$, there is an equitable partition $V = V_1\cup V_2\cup \cdots \cup V_K$, with $1/\varepsilon < K < \varepsilon^{-C}$, such that all but at most $\varepsilon{K\choose 4}$ quadruples of parts $\{V_{i_1},V_{i_2},V_{i_3},V_{i_4}\}$ have same-type transversals.
\end{theorem}

For small graphs $G = (V,E)$ with $|V(G)| = K$, we can compute $\lcr(G)$ as follows.  We generate ${K\choose 3}$ polynomials $f_1,f_2,\ldots, f_{{K\choose 3}} \in \RR[x_1,\ldots, x_K,y_1,\ldots, y_K]$, where for $1 \leq i_1  < i_2 < i_3\leq K$ and $\rk(i_1,i_2,i_3) = j$, we have

$$f_j = \hspace{.05cm} \det\left(\begin{array}{ccc}
                                                           1 & 1 &  1 \\
                                                           x_{i_1} & x_{i_2} & x_{i_3} \\
                                                           y_{i_1} & y_{i_2} & y_{i_3}
                                                         \end{array}\right).$$

Fix an abstract order type $\chi^{\ast}\in \{+1,-1\}^{K\choose 3}$, and let $j_1,\ldots,j_r$ be the indices for which $\chi^{\ast}_{j_{\ell}} = +1$, and let $j'_1,\ldots, j'_s$ be the indices for which $\chi^{\ast}_{j'_{\ell}} = -1$.  In order to decide if $\chi^{\ast}$ is realizable, we need to see if there are real solutions to the polynomial system

$$f_{j_1} > 0, \ldots, f_{j_r} > 0 \hspace{1cm} f_{j'_1}  < 0 ,\ldots,  f_{j'_s} < 0.$$

\noindent This is a special case of the satisfiability problem in the existential theory of the reals (see \cite{BLS}).  By an algorithm of Basu, Pollack, and Roy \cite{BPR}, we can decide if the polynomial system above has real solutions in $2^{O(K\log K)}$ time.  Moreover, if there are solutions, the algorithm will output a solution $(x_1,\ldots, x_K,y_1,\ldots, y_K)$, where each coordinate uses at most $2^{O(K \log K)}$ bits.  Hence if $\chi^{\ast}$ is realizable, we obtain a point set $V = \{v_1,\ldots , v_K\}$ in the plane that realizes $\chi^{\ast}$, and each point has at most $2^{O(K\log K)}$ bits.

If we do obtain such a point set $V$, we then compute the minimum number of pairs of crossing edges over all straight-line drawings of $G$ which uses $V$ as its vertex set.  This can be done in $2^{O(K\log K)}$ time.   By repeating the procedure above over all $2^{{K\choose 3}}$ abstract order types $\chi^{\ast}$, we have the following.

\begin{lemma}\label{brute}
Given a graph $G$ on $K$ vertices, we can find a straight-line drawing of $G$ with $\lcr(G)$ pairs of crossing edges in $2^{O(K^3)}$ time.

\end{lemma}

\noindent  See \cite{AiAK,AiK} for an alternative heuristic method for computing $\lcr(G)$.

\section{Cut-distance and the Frieze--Kannan regularity lemma}\label{sect2}

An \emph{edge weighted graph} $G = (V,E)$ is a graph with weights $w_G(uv) \in [0,1]$ associated with each edge $uv \in E(G)$.  For convenience, set $w_G(uv) = 0$ if $uv \not\in E(G)$.  For $S,T \subset V(G)$, we define

$$e_G(S,T) = \sum\limits_{u \in S, v\in T}w_G(uv).$$

\noindent Note that if the sets $S$ and $T$ have a nonempty intersection, the weights of the edges running in $S\cap T$ are counted twice.  Let $G$ and $G'$ be two edge weighted labeled graphs with the same vertex set $V = \{v_1,\ldots, v_n\}$.  The \emph{cut-distance} between $G$ and $G'$ is defined as

$$d(G,G') = \max\limits_{S,T\subset V} \left|e_G(S,T) - e_{G'}(S,T)\right|.$$

\noindent Hence, the cut-distance between two labeled graphs measures how different the two graphs are when considering the size of various cuts.  This concept has played a crucial role in the work of Frieze and Kannan \cite{FK} on efficient approximation algorithms for dense graphs.  See \cite{BCL} and the book \cite{Lov} for more results on cut-distance.

We generalize the concept of crossing numbers to edge weighted graphs as follows.  Let $\mathcal{D}$ be a straight-line drawing of $G$ in the plane, and let $X_\mathcal{D} \subset {E(G) \choose 2}$ denote the set of pairs of crossing edges in the drawing.  The \emph{rectilinear crossing number} of the \emph{edge-weighted graph} $G$ is defined as

$$\lcr(G) = \min\limits_{\mathcal{D}} \sum\limits_{(uv,st) \in X_\mathcal{D}}w_G(uv)\cdot w_G(st), $$

\noindent where the minimum is taken over all straight-line drawings of $G$.  Thus for any unweighted graph $G = (V,E)$, we can assign weights $w_G(uv) = 1$ for $uv \in E(G)$ and $w_G(uv) = 0$ for $uv \not\in E(G)$ so that the definition of $\lcr(G)$ remains consistent.  By copying the proof of Lemma \ref{brute} almost verbatim, we have the following lemma.

\begin{lemma}\label{brute2}
Let $G$ be an edge weighted graph on $K$ vertices, where the weight of each edge uses at most $B$ bits.  Then we can find a straight-line drawing of $G$ with $\lcr(G)$ weighted edge crossings in $2^{O(K^3)}B^2$ time.

\end{lemma}

Another key ingredient used in the proof of Theorem \ref{main} is a variant of Szemer\'edi's regularity lemma developed by Frieze and Kannan.  Szemer\'edi's regularity lemma~\cite{szemeredi} is one of the most powerful tools in modern combinatorics and gives a rough structural characterization of all graphs.  According to the lemma, for every $\varepsilon>0$ there is $K=K(\varepsilon)$ such that every graph has an equitable vertex partition into at most $K$ parts such that all but at most an $\varepsilon$ fraction of the pairs of parts behave ``regularly."\footnote{For a pair $(V_i,V_j)$ of vertex subsets, the density $d(V_i,V_j)$ is defined as $\frac{e_G(V_i,V_j)}{|V_i||V_j|}$. The pair $(V_i,V_j)$ is called $\varepsilon$-regular if for all $V_i' \subset V_i$ and $V_j' \subset V_j$ with $|V_i'| \geq \varepsilon |V_i|$ and $|V_j'| \geq \varepsilon |V_j|$, we have $|d(V'_i,V'_j)-d(V_i,V_j)| \leq \varepsilon$.} The dependence of $K$ on $1/\varepsilon$ is notoriously strong. It follows from the proof that $K(\varepsilon)$ may be taken to be an exponential tower of twos of height $\varepsilon^{-O(1)}$. Gowers \cite{Go97} used a probabilistic construction to show that such an enormous bound is indeed necessary.  This is quite unfortunate, because in algorithmic applications of the regularity lemma this parameter typically has a negative impact on the efficiency.   Consult \cite{CoFo}, \cite{MoSh}, \cite{FoLo} for other proofs that improve on various aspects of the result.

Frieze and Kannan \cite{FK} developed a weaker notion of regularity which is sufficient for certain algorithmic applications, and for which the dependence on the approximation parameter $\varepsilon$ is much better.  Let $\varepsilon > 0$ and let $G = (V,E)$ be a graph on $n$ vertices.  An equitable partition $\mathcal{P}:V = V_1\cup \cdots \cup V_K$ is said to be \emph{$\varepsilon$-Frieze-Kannan-regular} if for all subsets $S,T \subset V(G)$, we have

$$\left|e_G(S,T) - \sum\limits_{1 \leq i, j\leq K}  e_G(V_i,V_j)\frac{|S\cap V_i||T\cap V_j|}{|V_i||V_j|}    \right| < \varepsilon n^2.$$

\noindent Frieze and Kannan \cite{FK} showed that for any $\varepsilon > 0$, every graph $G = (V,E)$ has an $\varepsilon$-Frieze-Kannan-regular partition with $K$ parts, where $1/\varepsilon < K < 2^{O(\varepsilon^{-2})}$.  Moreover, such a partition can be found in randomized $O(n^2)$-time.  For the algorithm we present in the next section, we will use the following more recent algorithmic version due to Dellamonica et al.

\begin{theorem}[\cite{DKM}]\label{FKlemma} There is an absolute constant $c$ such that the following holds.  For each $\varepsilon > 0$ and for any graph $G = (V,E)$ on $n$ vertices, there is a deterministic algorithm which finds an $\varepsilon$-Frieze-Kannan-regular partition on $V$ with at most $2^{\varepsilon^{-c}}$ parts, and runs in $2^{2^{\varepsilon^{-c}}}n^2$-time.

\end{theorem}

Given an $n$-vertex graph $G = (V,E)$, let $\mathcal{P}:V = V_1\cup \cdots \cup V_K$ be an $\varepsilon$-Frieze-Kannan-regular partition obtained from Theorem \ref{FKlemma}.  We now define two edge-weighted graphs $G/\mathcal{P}$ and $G_\mathcal{P}$ as follows.  Let $G/\mathcal{P}$ be the edge-weighted graph on the vertex set $\{1,\ldots, K\}$ and with edge weights $$w_{G/\mathcal{P}}(ij) = \frac{e_G(V_i,V_j)}{(n/K)^2} \hspace{1cm} 1\leq i \neq j \leq K.$$   Let $G_\mathcal{P}$ be an edge-weighted graph with vertex set $V = V(G)$, and with edge weights

$$ w_{G_\mathcal{P}}(uv)  = \left\{\begin{array}{cl}
                                                    \frac{e_G(V_i,V_j)}{(n/K)^2} & \textnormal{if $u \in V_i, v\in V_j$, $1 \leq i \neq j \leq K$;} \\\\
                                                    0 & \textnormal{if $u,v \in V_i$, $1 \leq i \leq K$.}
                                                  \end{array}\right. $$

Thus, the Frieze--Kannan regularity lemma says that $d(G,G_\mathcal{P}) < \varepsilon n^2$, which implies that $G/\mathcal{P}$ is a small graph that gives a good approximation of $G$. We now prove the following lemmas which establish a relationship between $\lcr(G), \lcr(G_{\mathcal{P}}),$ and $\lcr(G/\mathcal{P})$.

\begin{lemma}\label{cross}
Let $\varepsilon \in (0,1/2)$ and let $G$ and $G'$ be two $n$-vertex edge-weighted graphs on the same vertex set $V$.  If $d(G,G') < \varepsilon n^2$, then we have

$$|\lcr(G) - \lcr(G')| \leq \varepsilon^{\frac{1}{4C}}n^4,$$

\noindent where $C$ is an absolute constant from Theorem \ref{partition}.

\end{lemma}

\begin{proof}
Consider a straight-line drawing $\mathcal{D}$ of $G = (V,E)$ in the plane such that if $X_{\mathcal{D}} \subset {E \choose 2}$ denotes the set of pairs of crossing edges in $\mathcal{D}$, we have

\begin{equation}\label{sum}\lcr(G) = \sum\limits_{(e_1,e_2) \in X_{\mathcal{D}}}w_G(e_1)w_G(e_2).\end{equation}

With slight abuse of notation, let $V$ be the point set in the plane representing the vertices of $G$ in the drawing $\mathcal{D}$. We can assume that $V$ is in general position.  With approximation parameter $\varepsilon^{1/(4C)}$, we apply Theorem \ref{partition} to the point set $V$ and obtain an equitable partition $V = V_1\cup \cdots\cup V_K$, where $K \leq \varepsilon^{-1/4}$, such that all  but at most $\varepsilon^{1/(4C)}{K\choose 4}$ quadruples of parts $(V_{i_1},V_{i_2},V_{i_3},V_{i_4})$ have same-type transversals.  Let $T \subset{[K]\choose 4}$ be the set of quadruples $(i_1,i_2,i_3,i_4)$ such that $(V_{i_1}, V_{i_2}, V_{i_3}, V_{i_4})$ has same type transversal and every such transversal is in convex position.  Then for each such quadruple, we can order the elements $(i_1,i_2,i_3,i_4) \in T$ so that every segment with one endpoint in $V_{i_1}$ and the other in $V_{i_2}$ crosses every segment with one endpoint in $V_{i_3}$ and the other in $V_{i_4}$.  Therefore, we have

\begin{equation}\label{lower}\lcr(G) \geq  \sum\limits_{(i_1,i_2,i_3,i_4) \in T}e_G(V_{i_1},V_{i_2})e_G(V_{i_3},V_{i_4}).\end{equation}

\noindent On the other hand, let us consider the drawing $\mathcal{D}'$ of $G'$ on the same point set $V = V_1\cup\cdots \cup V_K$.  We say that the quadruple $(v_1,v_2,v_3,v_4) \in {V\choose 4}$ is \emph{bad} if two members lie in a single part $V_j$, or if all four members lie in distinct parts $V_{i_1},V_{i_2},V_{i_3},V_{i_4}$ such that $(V_{i_1},V_{i_2},V_{i_3},V_{i_4})$ does not have same-type transversals.  By Theorem \ref{partition}, we have at most

$$K{\lceil n/K\rceil\choose 2}{n\choose 2} + \varepsilon^{\frac{1}{4C}} {K\choose 4}\left\lceil\frac{n}{K}\right\rceil^4 \leq \frac{n^4}{4K}  + Kn^2 + \varepsilon^{\frac{1}{4C}}{n\choose 4} \leq 2\varepsilon^{\frac{1}{4C}}{n\choose 4} ,$$

\noindent bad quadruples.  Since each edge has weight at most one, we have

$$\lcr(G')\leq \sum\limits_{(i_1,i_2,i_3,i_4) \in T}e_{G'}(V_{i_1},V_{i_2})e_{G'}(V_{i_3},V_{i_4}) + 2\varepsilon^{\frac{1}{4C}}{n\choose 4}.$$

\noindent Since $d(G,G') < \varepsilon n^2$, and by (\ref{lower}), we have

$$\begin{array}{ccl}
    \lcr(G') & \leq & \sum\limits_{(i_1,i_2,i_3,i_4) \in T}e_{G'}(V_{i_1},V_{i_2})e_{G'}(V_{i_3},V_{i_4}) + 2\varepsilon^{\frac{1}{4C}}{n\choose 4} \\\\
      & \leq &  \sum\limits_{(i_1,i_2,i_3,i_4) \in T}(e_{G}(V_{i_1},V_{i_2}) + \varepsilon n^2)(e_{G}(V_{i_3},V_{i_4}) + \varepsilon n^2) +2\varepsilon^{\frac{1}{4C}}{n\choose 4}\\\\
       &   \leq & \lcr(G) + \frac{\varepsilon^{1/2}n^4}{2} + \varepsilon\frac{n^4}{4!}+ 2\varepsilon^{\frac{1}{4C}}{n\choose 4}  \\\\
        & \leq &  \lcr(G) +  \varepsilon^{\frac{1}{4C}}n^4.
  \end{array}$$

\noindent The last inequality follows from the fact that $C$ is a sufficiently large constant.  A symmetric argument also shows that $\lcr(G) \leq \lcr(G') +  \varepsilon^{\frac{1}{4C}}n^4$, and the statement follows.\end{proof}

Let $G$ be an edge-weighted graph on the vertex set $V = \{1,\ldots, K\}$ with weights $w_G(i,j)$.  The blow-up $G[m]$ of $G$ is the edge-weighted graph obtained from $G$ by replacing each vertex $i$ by an independent set $U_i$ of order $m$, and each edge between $U_i$ and $U_j$ has weight $w_G(i,j)$ for $i\neq j$.

\begin{lemma}\label{blowup}

Let $G$ and $G[m]$ be described as above.  Then

$$ 0 \leq \lcr(G[m]) - m^4\lcr(G) \leq K^3m^4.$$

\end{lemma}

\begin{proof}  We start by proving the second inequality first.   Fix a drawing $\mathcal{D}$ of $G$ such that if $X$ denotes the set of pairs of crossing edges in $\mathcal{D}$, we have

$$\sum\limits_{(e_1,e_2) \in X}w_G(e_1)w_G(e_2) = \lcr(G).$$

\noindent Let $V$ be the point set in the plane representing the vertices of $G$ in the drawing.  We can assume that $V$ is in general position. We draw the blow-up graph $G[m]$ as follows.  For each point $v \in V$ in the plane, we choose a very small $\delta$ and add $m-1$ points in the disk centered at $v$ with radius $\delta$.  These points will represent $U_v$. By choosing $\delta$ sufficiently small, every quadruple of parts $(U_{i_1},U_{i_2},U_{i_3},U_{i_4})$ will have same-type transversals.  Moreover, we can do this so that the resulting point set is in general position.  Finally if $uv \in E(G)$, we draw all edges between the point sets $U_u$ and $U_v$.   Let $X_m$ denote the set of pairs of crossing edges in our drawing of $G[m]$.

Set $U = U_1\cup \cdots \cup U_K$.  We say that the quadruple $(u_1,u_2,u_3,u_4)$ of points in $U$ is bad if two of its members lie in a single part $U_i$.  Hence the number of bad quadruples in $U$ is at most $K{m\choose 2}{Km\choose 2}$.  Since each edge has weight at most one, we have

$$\begin{array}{ccl}
    \lcr(G[m]) & \leq  & \sum\limits_{(e_1,e_2) \in X_m} w_{G[m]}(e_1)w_{G[m]}(e_2) \\\\
      & \leq & m^4\lcr(G) + K{m\choose 2}{Km\choose 2}\\\\
      & \leq &   m^4\lcr(G) + K^3m^4.
  \end{array}$$

\noindent On the other hand, now consider a drawing $\mathcal{D}'$ of $G[m]$ such that if $X'$ denotes the set of pairs of crossing edges in $\mathcal{D}'$, we have

 $$\lcr(G[m]) = \sum\limits_{(e_1,e_2) \in X'}w_{G[m]}(e_1)w_{G[m]}(e_2).$$

\noindent Let $V(G[m]) = U_1\cup\cdots \cup U_K$.   By selecting one point from each $U_i$, we obtain a drawing of $G$ which has at least $\lcr(G)$ weighted pairs of crossing edges.  Summing over all of these $m^K$ distinct drawings of $G$, each weighted crossing appears $m^{K-4}$ times.  Therefore,

$$\lcr(G[m]) \geq \lcr(G)\cdot m^K/m^{K-4} = m^4\lcr(G).$$ This completes the proof.\end{proof}

\section{Proof of Theorem \ref{main}}\label{algorithm}

\noindent\textbf{The algorithm.}  Input: Let $G$ be a graph with vertex set $V = \{v_1,v_2,\ldots, v_n\}$.

\begin{enumerate}

\item Set $\varepsilon = (\log\log n)^{\frac{-1}{2c}}$, where $c$ is defined in Theorem \ref{FKlemma}.  We apply Theorem \ref{FKlemma} to $G$ with approximation parameter $\varepsilon$, and obtain an equitable partition $\mathcal{P}: V = V_1\cup \cdots \cup V_K$ on our vertex set with the desired properties such that $1/\varepsilon < K < 2^{\varepsilon^{-c}} = 2^{\sqrt{\log\log n}}$.  This can be done deterministically in $n^{2 + o(1)}$-time using the algorithm of Dellamonica et al. \cite{DKM}.

\item  Let $G/\mathcal{P}$ be the edge-weighted graph on the vertex set $\{1,\ldots, K\}$ with edge weights $w_{G/\mathcal{P}}(ij) =  \frac{e_G(V_i,V_j)}{|V_i||V_j|}$.  Using Lemma \ref{brute2}, we can find a drawing of $G/\mathcal{P}$ with $\lcr(G/\mathcal{P})$ weighted pairs of crossing edges.  Let $U = \{u_1,\ldots, u_K\}$ be the point set for such a drawing where each point uses at most $2^{O(K\log K)}$ bits.  This can be done in $2^{O(K^3)} = n^{o(1)}$ time.

\item We draw $G = (V,E)$ as follows.  Let $L$ be the set of lines spanned by $U$, and let $\delta$ be the minimum positive distance\footnote{The distance between a point and a line in the plane is the length of the line segment which joins the point to the line and is perpendicular to the line.} between the points in $U$ and the lines in $L$.  Note that $\delta$ uses at most  $2^{O(K\log K)}$ bits since the line spanned by any two points in $U$ will have the form $y = m_0x + b_0$, where $m_0$ and $b_0$ uses at most $2^{O(K\log K)}$ bits.  Therefore, the distance between a point in $U$ and the line $y = m_0x + b_0$ will use at most $2^{O(K\log K)}$ bits.    Set $D(i,\delta/10)$ to be the disk centered at $u_i$ with radius $\delta/10$.  We place the points of $V_i$ in $D(i,\delta/10)$ so that the point set $V_1\cup \cdots \cup V_K$ is in general position, and each point uses at most $2^{O(K\log K)} < O(n)$ bits.  Notice that every quadruple of parts $(V_{i_1},V_{i_2},V_{i_3},V_{i_4})$ has same-type transversals.  We then draw all edges of $G$ on this point set.  This can be done in $O(n^2)$ time.

\item Return: the drawing of $G$.

\end{enumerate}

\noindent The total running time for the algorithm above is $n^{2 + o(1)}$.

Let $\mathcal{D}$ be the drawing of $G = (V,E)$ obtained from the algorithm above, where $V = \{v_1,\ldots, v_n \} \subset \RR^2$, and let $X$ denote the set of pairs of crossing edges in $\mathcal{D}$.  We say that the quadruple of points $\{v_{i_1},v_{i_2},v_{i_3},v_{i_4}\}$ in $V$ is bad if two of its members lie in a single disk $D(j,\delta/10)$.  Hence there are at most $n^4/(2K)$ bad quadruples.  Therefore
\begin{equation}\label{one}
|X| \leq \left(\frac{n}{K}\right)^4\lcr(G/\mathcal{P}) + \frac{n^4}{2K}.\end{equation}

\noindent Just as above, let $G_\mathcal{P}$ be the edge weighted graph with vertex set $V$ (same as $G$), with edge weights $w_{G_\mathcal{P}}(uv)  = e_G(V_i,V_j)/(n/K)^2$, if $u \in V_i$, $v \in V_j$ and $i\neq j$, and $w_{G_\mathcal{P}}(uv) = 0$ otherwise.  Since $G_{\mathcal{P}}$ is an $(n/K)$-blow-up of $G/\mathcal{P}$, by the proof of Lemma \ref{blowup}, we have

\begin{equation}\label{two}\left(\frac{n}{K}\right)^4\lcr(G/\mathcal{P}) \leq \lcr(G_{\mathcal{P}}).\end{equation}

\noindent Since Theorem \ref{FKlemma} implies that the cut-distance between $G$ and $G_{\mathcal{P}}$ satisfies $d(G,G_{\mathcal{P}}) < \varepsilon n^2$, Lemma \ref{cross} implies that

\begin{equation}\label{three}\lcr(G_{\mathcal{P}}) \leq \lcr(G) + \varepsilon^{\frac{1}{4C}}n^4.\end{equation}

\noindent Putting together (\ref{one}), (\ref{two}), and (\ref{three}), shows that

$$|X| < \lcr(G) + O\left(\frac{n^4}{(\log\log n)^{\delta}}\right),$$

\noindent  where $\delta$ is an absolute constant.  This completes the proof of Theorem \ref{main}.

\section{The rectilinear crossing number of quasi-random graphs}\label{quasisect}

\emph{Proof of Theorem \ref{quasithm}.}  Let $\mathcal{D}$ be a straight-line drawing of $G_n$ in the plane with exactly $\lcr(G_n)$ edge crossings, and let $V = \{v_1,v_2,\ldots, v_n\}$ be the point set in the plane that represents the vertices of $G_n$ in the drawing.  Without loss of generality, we can assume no three members of $V$ are collinear, no two members of $V$ share the same $x$-coordinate, and $V$ is ordered by increasing $x$-coordinate.

Set $\varepsilon  = n^{-1/2C}$, where $C$ is defined in Theorem \ref{partition}.  By Theorem \ref{partition}, there is an equitable partition $V = V_1\cup \cdots\cup V_K$ into $K$ parts, where $K \leq \varepsilon^{-C} = \sqrt{n}$, such that all but at most $\varepsilon{K\choose 4}$ quadruples $(V_{i_1},V_{i_2},V_{i_3},V_{i_4})$ of parts have same-type transversals.  Let $Q \subset {V\choose 4}$ be the set of quadruples in $V$ that are in convex position.  We say that a quadruple $(v_{i_1},v_{i_2},v_{i_3},v_{i_4}) \in Q$ is \emph{bad} if two of its members lie in a single part $V_j$, or if they lie in distinct parts of $(V_{j_1},V_{j_2},V_{j_3}, V_{j_4})$ the does not have same-type transversals.  Hence the number of bad quadruples in $Q$ is at most

$$K{n/K\choose 2}{n\choose 2} +  \varepsilon{K\choose 4}\left\lceil\frac{n}{K}\right\rceil^4 \leq \frac{n^4}{4K} + \varepsilon {n\choose 4}.$$

Let $T$ denote the number of quadruples of parts $(V_{i_1},V_{i_2},V_{i_3},V_{i_4})$, where each such quadruple $(V_{i_1},V_{i_2},V_{i_3},V_{i_4})$ has same-type transversals and each such transversal is in convex position.  Then we have

$$T \cdot \left(\frac{n}{K}\right)^4 \geq |Q| - \left(\frac{n^4}{4K} + \varepsilon {n\choose 4}\right)\geq \lcr(K_n) - \left(\frac{n^4}{4K} + \varepsilon {n\choose 4}\right).$$

\noindent Since $$\lcr(G_n) \geq T\left(p\lfloor n/K\rfloor^2 - o(n^2)\right)^2 = Tp^2\lfloor n/K\rfloor^4 - o(n^4),$$ this implies

$$\lcr(G_n) \geq p^2\lcr(K_n) - p^2\left(\frac{n^4}{4K} + \varepsilon {n\choose 4}\right) - o(n^4) = p^2\lcr(K_n) - o(n^4).$$

On the other hand, drawing $G_n$ in the plane by placing its vertices on a point set that minimizes the number of quadruples in convex position, one can follow the arguments above to show that

 $$\lcr(G_n) \leq p^2\lcr(K_n) + o(n^4).$$  This completes the proof of Theorem \ref{quasithm}.\qed

\section{Concluding remarks}

Pach et al.~\cite{shah} introduced the following alternative notion of crossing number.  For any positive integer $k\geq 1$, the \emph{geometric $k$-planar crossing number} of $G$, denoted by $\lcr_k$, is the minimum number of crossings between edges of the same color over all $k$-edge-colorings of $G$ and all straight-line drawings of $G$. By following the proof of Theorem~\ref{main} almost verbatim, we have the following theorem.  We note that one needs to slightly modify the proof of Lemma \ref{cross}, by coloring the edges of $G'$ between parts $V_{i_1}$ and $V_{i_2}$ ($V_{i_3}$ and $V_{i_4}$), so that the number of edges in color $i$ between the two parts in $G'$ is roughly the same as the number of edges in color $i$ between the two parts in $G$.

\begin{theorem}

Let $k\geq 1$ be a fixed constant.  Given any $n$-vertex graph $G$, there is a deterministic $n^{2 + o(1)}$-time algorithm that finds a straight-line drawing of $G$ in the plane, and a $k$-coloring of the edges in $G$, such that the number of monochromatic pairs of crossing edges in the drawing is at most $\lcr_k(G) + O(n^4/(\log\log n)^{\delta})$, where $\delta$ is an absolute constant.

\end{theorem}

\noindent We suspect that Theorem \ref{main} also holds for other crossing number variants.

Let us also remark that the rectilinear crossing number is a testable parameter, which means that there is a constant time randomized algorithm for approximating the rectilinear crossing number. More precisely, for each $\epsilon>0$ there is $t=t(\epsilon)>0$ such that the following holds. If $G$ is a graph on $n$ vertices, by sampling a random induced subgraph $H$ of G on $t$ vertices, we can approximate with probability of success at least $.99$ the rectilinear crossing number of $G$ with error at most $\epsilon n^4$. We do this by noting that the random sample $H$ is, with probability at least $.99$, close in cut-distance to $G$ (see the Lov\'asz book \cite{Lov} for details). By Lemma \ref{cross}, if they are close in cut-distance, we get that $\lcr(G)$ is within $\epsilon n^4$ of $\lcr(H)\frac{n^4}{t^4}$.

\end{document}